\documentclass[12pt]{article}
\usepackage[english]{babel}

\usepackage{amsmath,amssymb,amsthm}

\usepackage[pdfpagemode=UseOutlines,
bookmarks=true, bookmarksopenlevel=3, bookmarksopen=true,
bookmarksnumbered=true, unicode=true,
pdffitwindow=true,pdfpagelayout=SinglePage,pdfhighlight=/P,
colorlinks=true,linkcolor=blue,citecolor=blue,urlcolor=blue]{hyperref}

\usepackage[left=2.5cm, right=2.5cm, top=2.5cm, bottom=2.5cm]{geometry}

\allowdisplaybreaks[2]

\theoremstyle{plain}
\newtheorem{proposition}{Proposition}

\theoremstyle{remark}
\newtheorem{remark}{Remark}

\theoremstyle{definition}
\newtheorem{definition}{Definition}

\title{3D consistency of negative flows}

\author{V.E.\:Adler\thanks{L.D.\:Landau Institute for Theoretical Physics, Akad. Semenova av. 1A, 142432, Chernogolovka, Russian Federation. E-mail: adler@itp.ac.ru}}

\date{July 11, 2024}

\begin{document}
\maketitle

\begin{abstract}
We study the 3D-consistency property for negative symmetries of KdV type equations. Its connection with the 3D-consistency of discrete equations is explained.

\medskip

\end{abstract}

\section{Introduction}

The paper is devoted to partial differential equations of general form
\begin{equation}\label{uxxz}
 u_{xxz}=g(u,u_x,u_{xx},u_{xxx},u_z,u_{xz};\alpha),
\end{equation}
compatible with integrable equations of Korteweg-de Vries (KdV) type
\begin{equation}\label{ut}
 u_t=f(u,u_x,u_{xx},u_{xxx}).
\end{equation}
We call (\ref{uxxz}) a negative symmetry for (\ref{ut}) because in many examples such equations are constructed using the recursion operator $R$ for (\ref{ut}) by the formula
\begin{equation}\label{neg}
 u_z = (R-\alpha)^{-1}(u_{t_0}).
\end{equation}
The seed symmetry $u_{t_0}$ is usually taken to be $u_x$, $u_t$, or simply 0. This yields expansions in the parameter $\alpha$
\begin{equation}\label{expansion}
 u_z= -\alpha^{-1}(u_{t_0}+\alpha^{-1}u_{t_1}+\alpha^{-2}u_{t_2}+\dots)
    = u_{t_{-1}}+\alpha u_{t_{-2}}+\alpha^2u_{t_{-3}}+\dots
\end{equation}
where $u_{t_n}=R^n(u_{t_0})$. Thus, the flow (\ref{neg}) serves as a generating function for the hierarchy of higher and negative symmetries of the equation (\ref{ut}), which explains its role in the theory. This approach to negative symmetries, for equations of the KdV type and also for such systems as the nonlinear Schr\"odinger equation, the Boussinesq equation, etc., was used, for example, in papers \cite{Kamchatnov_Pavlov_2002, Aratyn_Gomes_Zimerman_2006, Adans_Franca_Gomes_Loboa_Zimerman_2023, Lou_Jia_2024, Adler_2024a, Adler_2024b}. It is worth noting that in many cases negative symmetries (\ref{neg}) are related with equations of independent interest. For example, the negative flow for KdV is reduced by additional substitutions to the famous Camassa--Holm equation \cite{Schiff_1998, Hone_1999}, negative symmetry for the nonlinear Schr\"odinger equation is equivalent to the Maxwell--Bloch system or a two-component analog of the Camassa--Holm equation \cite{Rogers_Schief_2002, Aratyn_Gomes_Zimerman_2006b}, some equations (\ref{uxxz}) admit reductions to sine-Gordon type equations that define hyperbolic symmetries for equations of KdV type \cite{Meshkov_Sokolov_2011}. An important application of negative symmetries is related with the construction of non-autonomous finite-dimensional reductions of Painlev\'e type: it turns out that stationary equations for symmetries from an additional subalgebra (the so-called string equations) are equivalent to stationary equations for a linear combination of a higher symmetry, a classical symmetry like the scaling or Galilean transformation, and a sum of an arbitrary number of negative flows (\ref{neg}) corresponding to different values of parameter $\alpha$ \cite{Orlov_Rauch-Wojciechowski_1993, Adler_Kolesnikov_2023, Adler_2024a}. In this regard, the question arises about the compatibility of such negative flows, which is discussed in this article.

In a typical situation, higher symmetries $u_{t_n}=R^n(u_{t_0})$ for $n\ge0$ determine pairwise commuting evolution differentiations, then from the first expansion (\ref{expansion}) it follows that the flows $\partial_{z_i}$ and $\partial_{z_j}$ corresponding to the values $\alpha_i$ and $\alpha_j$ also commute with each other. It also follows that $\partial_{t_n}$ are commutative for $n<0$ (which is not easy to verify directly, since these flows are non-local).

In general, we need a definition of the compatibility of the flows $\partial_{z_i}$ and $\partial_{z_j}$ that does not depend on the method of their construction. We will show that such a definition must include another equation involving both variables $z_i$ and $z_j$. As a simplified illustration, let us consider an example of a 3D-consistent triple of hyperbolic equations (Ferapontov \cite{Ferapontov_1997}):
\begin{equation}\label{Fer}
 u_{xy}= \sinh u\sqrt{1+u^2_y},\quad
 u_{xz}= \cosh u\sqrt{1+u^2_z\strut},\quad
 u_{yz}= \sqrt{1+u^2_y}\sqrt{1+u^2_z\strut}.
\end{equation}
Here, consistency means the equality of cross derivatives for each pair of equations, provided that the third one is satisfied. For example, for the first two equations the calculations give
\begin{equation}\label{Fer.xyz}
 (u_{xy})_z-(u_{xz})_y= \left(u_{yz}-\sqrt{1+u^2_y}\sqrt{1+u^2_z\strut}\right)
 \left(\frac{u_z\cosh u}{\sqrt{1+u^2_z}} -\frac{u_y\sinh u}{\sqrt{1+u^2_y}}\right)
\end{equation}
which vanishes if the third equation is satisfied. Similar relations with the factorized right-hand side hold for the other two pairs. As a result, each equation of the triple is restored from the other two, if in the equality for cross derivatives the factors with lower derivatives are discarded. For a pair of equations of type (\ref{uxxz}), the situation is less symmetrical, since the additional third equation is of a different type, but the general scheme remains the same. Section \ref{s:def} defines the 3D-consistency property for equations (\ref{uxxz}) and explains the algorithm for checking it for given equations. Section \ref{s:examples} contains several simple examples corresponding to negative symmetries for equations of the KdV type.

Section \ref{s:lattice} describes the connection of the equations under consideration with the discrete case which has been studied much better. Recall that the concept of 3D-consistency for discrete equations on a square lattice (quad-equations) has been studied in detail in many publications, see e.g. \cite{Nijhoff_Walker_2001, Adler_Bobenko_Suris_2003, Adler_Bobenko_Suris_2009}; some generalizations to higher order equations considered in \cite{Adler_Postnikov_2014, Xenitidis_2019} can be interpreted as difference analogues of (\ref{uxxz}). Continuous symmetries for quad-equations (references are given in section \ref{s:lattice}) are divided into two types: dressing chains and Volterra-type chains,
\begin{equation}\label{i.chains}
 a(u_n,u_{n,x},u_{n+1},u_{n+1,x};\alpha)=0,\qquad u_{n,z} = g(u_{n-1},u_n,u_{n+1})
\end{equation}
where the subscript $n$ denotes a disrete variable (one of many on a multidimensional lattice). If the pair (\ref{i.chains}) is consistent, then eliminating the variables $u_{n\pm1}$ and their derivatives leads to an equation of type (\ref{uxxz}) for $u_n$, which was noted in the work of Yamilov \cite{Yamilov_1990}. Symbolically we can write
\[
 \text{negative flow}=\frac{\text{Volterra chain}}{\text{dressing chain}},
\]
that is, an equation of type (\ref{uxxz}) is introduced as a ``quotient equation'' for the Volterra type chain modulo the dressing one. In this setting, the 3D-consistency of negative flows follows from the commutativity of the base flows on the discrete 3D-consistent lattice. This method is an alternative to the formula (\ref{neg}) with a recursion operator and, generally speaking, may lead to different answers. Of the equations discussed in section \ref{s:lattice}, the example of the Krichever--Novikov equation is of particular interest. For it, the recursion operator is of the fourth order, and not the second, as for other equations of the KdV type. The general negative symmetry constructed by this recursion operator is more complicated compared to (\ref{uxxz}), but it was shown in \cite{Adler_2024b} that it allows a non-trivial reduction to an equation of this type. We now show that lattice equations provide a simpler way to derive this special negative symmetry.

\section{Definition of 3D-consistency}\label{s:def}

We will say that (\ref{uxxz}) is a negative symmetry for (\ref{ut}) if differentiating (\ref{uxxz}) due to (\ref{ut}) is a differential consequence of the equation (\ref{uxxz}) itself, that is, a relation holds
\begin{equation}\label{zt-cons}
 \left(D^2D_z-\frac{\partial g}{\partial u}-\frac{\partial g}{\partial u_x}D-\dots-\frac{\partial g}{\partial u_{xz}}DD_z\right)(f) 
 = A(u_{xxz}-g)
\end{equation}
where $D=D_x$ is the total $x$-derivative and $A=a_0D^3+a_1D^2+a_2D+a_3$ is some differential operator. In practice, verifying the equality (\ref{zt-cons}) for a given pair of equations comes down to eliminating the derivatives $\partial^n_x(u_z)=D^{n-2}(g)$ with $n\ge2$ from the left hand side, which is an algorithmic calculation. A more difficult question is the following: suppose that the equation (\ref{ut}) admits different negative symmetries (or a family depending on a continuous parameter), are they compatible with each other? As explained in the Introduction, this can be expected if negative symmetries are constructed by formula (\ref{neg}) using the recursion operator. However, even in this case, an independent algorithm for checking compatibility is required. Let us accept the following definition.

\begin{definition}\label{def:cons}
Let the derivations $D_{z_i}$ be defined by equations
\begin{equation}\label{uxxzi}
 u_{xxz_i}=g_i(u,u_x,u_{xx},u_{xxx},u_{z_i},u_{xz_i}),\quad i\in I,
\end{equation}
for some set of indices $I$. We say that equations (\ref{uxxzi}) are 3D-consistent if there exist additional equations
\begin{equation}\label{uzz}
 u_{z_iz_j}=g_{ij}(u,u_x,u_{xx},u_{z_i},u_{xz_i},u_{z_j},u_{xz_j}),\quad i\ne j,
\end{equation}
such that $g_{ij}=g_{ji}$ and the following equalities for pairwise distinct $i,j,k\in I$
\begin{gather}
\label{zz-cons1}
 D_{z_i}(g_j)=D_{z_j}(g_i)=D^2(g_{ij}),\\
\label{zz-cons2}
 D_{z_i}(g_{jk})=D_{z_j}(g_{ik})=D_{z_k}(g_{ij})
\end{gather}
hold identically by virtue of the equations (\ref{uxxzi}), (\ref{uzz}) and their differential consequences $u_{xxxz_i}=D(g_i)$, $u_{xz_iz_j}=D(g_{ij})$. 
\end{definition}

The identities (\ref{zz-cons1}) and (\ref{zz-cons2}) imply the coincidence of cross derivatives of an arbitrary order, which guarantees the existence of local solutions of general form that simultaneously satisfy (\ref{uxxzi}) and (\ref{uzz}) for all $i,j$.

At first glance, the Definition \ref{def:cons} is not constructive, since the equations (\ref{uzz}) are unknown in advance. However, if such equations exist, then they can be restored from the given equations (\ref{uxxzi}) by direct, albeit tedious, calculations. At the first step we obtain a relation of the form
\[
 0=D_{z_i}(g_j)-D_{z_j}(g_i)= P_{ij}(u,u_x,u_{xx},u_{xxx},u_{z_i},u_{xz_i},u_{z_j},u_{xz_j},u_{z_iz_j},u_{xz_iz_j}),
\]
where the derivatives $u_{xxxz_i}$, $u_{xxxz_j}$, $u_{xxz_i}$ and $u_{xxz_j}$ are eliminated from the right hand side due to (\ref{uxxzi}). By solving this equality with respect to $u_{xz_iz_j}$ (here we additionally assume that this derivative does not cancel identically), we obtain an equation of the form
\begin{equation}\label{uxzz}
 u_{xz_iz_j}=h_{ij}(u,u_x,u_{xx},u_{xxx},u_{z_i},u_{xz_i},u_{z_j},u_{xz_j},u_{z_iz_j})
\end{equation}
which should be a consequence (\ref{uzz}) to ensure consistency. At the second step, we analyze the condition
\[
 0=D(h_{ij})-D_{z_j}(g_i)
  =Q_{ij}(u,u_x,u_{xx},u_{xxx},u_{xxxx},u_{z_i},u_{xz_i},u_{z_j},u_{xz_j},u_{z_iz_j}),
\]
where, again, $u_{xxxz_i}$, $u_{xxxz_j}$, $u_{xxz_i}$ and $u_{xxz_j}$ are eliminated by (\ref{uxxzi}), and $u_{xz_iz_j}$ is eliminated in virtue of obtained equation (\ref{uxzz}). By solving this equality with respect to $u_{z_iz_j}$ (again, assuming that this derivative has not been vanished) we obtain the desired equation (\ref{uzz}). After this, all that remains is to check whether the equalities $D(g_{ij})=h_{ij}$ and $D_{z_i}(g_{jk})=D_{z_j}(g_{ik})$ turn into identities, which is done by direct calculations. Let us illustrate this scheme with several specific equations.

\section{Examples}\label{s:examples}

\subsection{Potential KdV equation}

The KdV equation
\begin{equation}\label{KdV}
 u_t=u_{xxx}-6uu_x
\end{equation}
admits the recursion operator $R=D^2-4u-2u_xD^{-1}$. For the convenience of further formulas, we replace $\alpha$ in the defining equation (\ref{neg}) with $-4\alpha$ and take $u_{t_0}=0$ as the seed symmetry (the choice $u_{t_0}=u_x$ is essentially the same, taking into account the integration constant in the term $u_xD^{-1}$). This gives the relation $R(u_z)=-4\alpha u_z$ which is equivalent to
\begin{equation}\label{KdV.uz}
 u_z=q_x,\quad q_{xxx}-4(u-\alpha)q_x-2u_xq=0.
\end{equation}
The order of second equation is lowered by integration with the factor $2q$:
\begin{equation}\label{KdV.qxx}
 2qq_{xx}-q^2_x-4(u-\alpha)q^2+4\beta=0,
\end{equation}
where $4\beta$ is the integration constant. This is a well-known equation for the resolvent of the Sturm--Liouville operator \cite{Gelfand_Dikii_1975}, with $\alpha$ playing the role of the spectral parameter. Eliminating the variable $q$, it is possible to obtain an equation of the form (\ref{uxxz}) for the variable $u$, but it is quite cumbersome. It is more convenient to make the substitution $2v_x=u$, $2v_z=q$, which leads to the potential KdV equation
\begin{equation}\label{KdV.vt}
 v_t=v_{xxx}-6v^2_x
\end{equation}
and the associated Camassa--Holm equation \cite{Schiff_1998, Hone_1999} as its negative symmetry:
\begin{equation}\label{KdV.vz}
 2v_zv_{xxz}-v^2_{xz}-4(2v_x-\alpha)v^2_z+\beta=0.
\end{equation}
The pair of equations (\ref{KdV.vt}) and (\ref{KdV.vz}) is consistent: one can check that if $\Phi$ is the left hand side of (\ref{KdV.vz}) then differentiating in virtue of (\ref{KdV.vt}) gives the identity
\[
 D_t(\Phi)= D^3(\Phi)-\frac{3v_{xz}}{v_z}D^2(\Phi) +3\left(\frac{v^2_{xz}}{v^2_z}-4v_x\right)D(\Phi).
\]
In the equation family (\ref{KdV.vz}), the main parameter is $\alpha$. We will show that 3D-consistency holds for equations corresponding to different values of $\alpha$, with no restrictions on the values of $\beta$. 

\begin{proposition}\label{pr:KdV}
Equations
\begin{equation}\label{KdV.zi}
 v_{xxz_i}=\frac{v^2_{xz_i}-\beta_i}{2v_{z_i}}+2(2v_x-\alpha_i)v_{z_i},
\end{equation}
are 3D-consistent with equations
\begin{equation}\label{KdV.vzz}
 v_{z_iz_j}=\frac{v_{z_i}v_{xz_j}-v_{z_j}v_{xz_i}}{\alpha_i-\alpha_j},\quad \alpha_i\ne\alpha_j.
\end{equation}
\end{proposition}

The proof is obtained by direct calculation. Let's use this example to demonstrate the main steps, including the procedure described above for deriving additional equations (\ref{KdV.vzz}). First, the equality $(v_{xxz_i})_{z_j}=(v_{xxz_j})_{z_i}$ brings to an equation of the form (\ref{uxzz}):
\begin{equation}\label{KdV.vxzz}
\begin{aligned}
 v_{xz_iz_j}
  &= \left(2(\alpha_i-\alpha_j)v_{z_i}v_{z_j}
     +\frac{\beta_jv_{z_i}}{2v_{z_j}}
     -\frac{\beta_iv_{z_j}}{2v_{z_i}}\right)
     \frac{v_{z_iz_j}}{v_{z_j}v_{xz_i}-v_{z_i}v_{xz_j}}\\
  &\qquad +\frac{1}{2}\left(\frac{v_{xz_i}}{v_{z_i}}+\frac{v_{xz_j}}{v_{z_j}}\right)v_{z_iz_j} +4v_{z_i}v_{z_j}.
\end{aligned}
\end{equation}
In the second step, the condition $(v_{xz_iz_j})_x=(v_{xxz_i})_{z_j}$ brings to the factorized equation 
\begin{multline*}
\qquad
 ((\alpha_i-\alpha_j)v_{z_iz_j}+v_{z_j}v_{xz_i}-v_{z_i}v_{xz_j})\times\\
 \times\frac{(v^2_{z_i}v^2_{xz_j}-v^2_{z_j}v^2_{xz_i}-4(\alpha_i-\alpha_j)v^2_{z_i}v^2_{z_j}+\beta_iv^2_{z_j}-\beta_jv^2_{z_i})}
  {(v_{z_j}v_{xz_i}-v_{z_i}v_{xz_j})^2}=0,
\qquad
\end{multline*}
where $v_{z_iz_j}$ is contained only in the first factor (compare with the example (\ref{Fer.xyz}) from the Introduction). Equating this factor to zero gives equation (\ref{KdV.vzz}). Next, we check that the equality $(v_{z_iz_j})_x=v_{xz_iz_j}$ is satisfied identically, that is, that the equation (\ref{KdV.vxzz}) is a consequence of equations (\ref{KdV.vzz}) and (\ref{KdV.zi}). More precisely, differentiation (\ref{KdV.vzz}) with respect to $x$ gives
\begin{equation}\label{KdV.vxzz'}
 v_{xz_iz_j}=2v_{z_i}v_{z_j}+\frac{1}{2(\alpha_i-\alpha_j)}
  \left(\frac{v_{z_i}}{v_{z_j}}(v^2_{xz_j}-\beta_j)-\frac{v_{z_j}}{v_{z_i}}(v^2_{xz_i}-\beta_i)\right),
\end{equation}
which coincides with (\ref{KdV.vxzz}) where $v_{z_iz_j}$ is replaced according to (\ref{KdV.vzz}). At the final step we check the fulfillment of the identities (\ref{zz-cons2}), that is $(v_{z_iz_j})_{z_k}=(v_{z_iz_k})_{z_j}$, which completes the proof of 3D-consistency.

\begin{remark}\label{rem:3Deq}
It should be noted that the equation (\ref{KdV.vzz}) is an independent three-dimensional integrable equation associated with the universal hydrodynamic Alonso--Shabat hierarchy \cite{Alonso_Shabat_2004}. It was shown in \cite{Adler_Shabat_2007} that the identity $(v_{z_iz_j})_{z_k}=(v_{z_iz_k})_{z_j}$ is fulfilled for this equation without taking the equations (\ref{KdV.zi}) into account. Equations (\ref{KdV.zi}) define a 2D reduction of this 3D equation which preserves the consistency property. The same is true also for the equation (\ref{S-KdV.uzz}) from the next section, which also was noted in \cite{Adler_Shabat_2007}. However, in general, the Definition \ref{def:cons} does not require that the identities (\ref{zz-cons2}) hold regardless of (\ref{uxxzi}).
\end{remark}

\subsection{Schwarzian KdV equation}

For the Schwarzian KdV equation
\begin{equation}\label{S-KdV}
 u_t= u_{xxx}-\frac{3u^2_{xx}}{2u_x},
\end{equation}
the recursion operator is \cite{Dorfman_1987, Wang_2002}
\[
 R= D^2 -\frac{2u_{xx}}{u_x}D +\frac{u_{xxx}}{u_x}-\frac{u^2_{xx}}{u^2_x}
   -u_xD^{-1}\cdot\left(\frac{u_{xxxx}}{u^2_x}-\frac{4u_{xx}u_{xxx}}{u^3_x}+\frac{3u^3_{xx}}{u^4_x}\right).
\] 
Deriving negative symmetry from the equation $R(u_z)=4\alpha u_z$ is staightforward and we will immediately give the answer.

\begin{proposition}\label{pr:S-KdV}
Equations
\begin{equation}\label{S-KdV.zi}
 u_{xxz_i}=\frac{u^2_{xz_i}-\beta_iu^2_x}{2u_{z_i}}+\frac{u_{xx}u_{xz_i}}{u_x}+2\alpha_iu_{z_i}
\end{equation}
are consistent with equation (\ref{S-KdV}) and 3D consistent with equations
\begin{equation}\label{S-KdV.uzz}
 u_{z_iz_j}=\frac{\alpha_iu_{z_i}u_{xz_j}-\alpha_ju_{z_j}u_{xz_i}}{(\alpha_i-\alpha_j)u_x},\quad 
 \alpha_i\ne\alpha_j.
\end{equation}
\end{proposition}

Equation (\ref{S-KdV.zi}) is slightly more complicated compared to (\ref{KdV.zi}), due to occurrence of $u_{xx}$ in the right hand side. However, computational verification of 3D-consistency follows the same scheme. As in the case of KdV, the identity $(u_{z_iz_j})_{z_k}=(u_{z_iz_k})_{z_j}$ is fulfilled without using equations (\ref{S-KdV.zi}); moreover, equation (\ref{S-KdV.uzz}) is even more symmetric than (\ref{KdV.vzz}) since in it variable $x$ is on equal footing with $z_i$ and $z_j$.

It is interesting to note that equation (\ref{S-KdV}) is consistent along with (\ref{S-KdV.zi}) also with each of the following hyperbolic equations:
\[
 u_{xz}=2u_x\sqrt{u_z},\quad u_{xz}=2uu_x,\quad u_{xz}=\frac{2uu_xu_z}{u^2+1},
\]
which can be interpreted as degenerate negative symmetries. It is easy to verify that the first of these equations determines special solutions of negative symmetry of the form (\ref{S-KdV.zi}) with $\alpha=\beta=0$, but the other two equations seem to have a different origin. We will not dwell on the question of their consistency with each other and with the flows (\ref{S-KdV.zi}). Many integrable evolution equations \cite{Meshkov_Sokolov_2011} have hyperbolic symmetries, although not all of them: in particular, they do not exist for the KdV (\ref{KdV}) and pot-KdV (\ref{KdV.vt}) equations.

\subsection{Dym equation}

The Dym equation
\begin{equation}\label{Dym.ut}
 u_t=u^3u_{xxx}
\end{equation}
admits the recursion operator (see e.g. \cite{Wang_2002}) 
\[
 R=u^2D^2-uu_xD+uu_{xx}+u^3u_{xxx}D^{-1}u^{-2} = u^3D^3uD^{-1}u^{-2}.
\]  
Equation for the negative symmetry $R(u_z)=\alpha u_z$ is equivalent to the system
\begin{equation}\label{Dym.uq}
 u_z=u^2q_x,\quad u^3(uq)_{xxx}-\alpha u^2q_x=0.
\end{equation}
The second equation admits an integrating factor $u^{-2}q$ and reduces to
\[
 2uq(uq)_{xx}-(uq)^2_x-\alpha q^2=\beta
\]
where $\beta$ is an integration constant. Similar to the KdV case, it is more convenient to use the potential form of the equation. The substitution
\begin{equation}\label{Dym.uqv}
 u=-1/v_x,\quad q=v_z
\end{equation}
turns (\ref{Dym.ut}) into equation 
\begin{equation}\label{Dym.vt}
 v_t= -\frac{v_{xxx}}{v^3_x}+\frac{3v^2_{xx}}{2v^4_x}
\end{equation}
(which is related to (\ref{S-KdV}) by the change $x\leftrightarrow v$) while equations for $q$ turn into a negative symmetry of the form (\ref{uxxz}). 

\begin{proposition}\label{pr:Dym}
Equations
\begin{equation}\label{Dym.vzi}
 2\frac{v_{z_i}}{v_x}\left(\frac{v_{z_i}}{v_x}\right)_{xx}
  =\left(\frac{v_{z_i}}{v_x}\right)^2_x+\alpha_iv^2_{z_i}+\beta_i
\end{equation}
are consistent with (\ref{Dym.vt}) and 3D-consistent with equations
\begin{equation}\label{Dym.vzz}
 v_{z_iz_j}= \frac{\alpha_iv_{z_j}v_{xz_i} -\alpha_jv_{z_i}v_{xz_j}}{(\alpha_i-\alpha_j)v_x},\quad 
 \alpha_i\ne\alpha_j.
\end{equation}
\end{proposition}

Notice that equations (\ref{Dym.vzz}) coincide with (\ref{S-KdV.uzz}) up to the change $\alpha\to1/\alpha$.

The equation $(R-\alpha)(u_z)=c_1u_x+c_2u_t$ with more general seed symmetry is still reduced to $R(u_z)=\alpha u_z$ if $\alpha\ne0$,
by suitable choice of antiderivative in the integral term of $R$ and additional transformation $\partial_z\mapsto \partial_z-\frac{c_1}{\alpha}\partial_x$. If $\alpha=0$ then we get the system
\[
 u_z=u^2q_x,\quad u^3(uq)_{xxx}=u_x
\]
instead of (\ref{Dym.uq}), which admits the integrating factor $u^{-3}$. The change (\ref{Dym.uqv}) brings to one more negative symmetry which complements the family (\ref{Dym.vzi}) and is consistent with it:
\[ 
 2(v_z/v_x)_{xx}=v_x^2+\beta.
\]
It can be combined with (\ref{Dym.vzi}) into a more general 3D-consistent family
\[
 2\left(\alpha_i\frac{v_{z_i}}{v_x}+\gamma_i\right)\left(\frac{v_{z_i}}{v_x}\right)_{xx}
  =\alpha_i\left(\frac{v_{z_i}}{v_x}\right)^2_x+(\alpha_iv_{z_i}+\gamma_iv_x)^2+\beta_i,
\]
with (\ref{Dym.vzz}) replaced by
\[
 v_{z_iz_j}= \frac{\alpha_iv_{z_j}v_{xz_i}-\alpha_jv_{z_i}v_{xz_j}
  +\gamma_j(v_xv_{xz_i}-v_{xx}v_{z_i})-\gamma_i(v_xv_{xz_j}-v_{xx}v_{z_j})}{(\alpha_i-\alpha_j)v_x}.
\]

\subsection{Potential mKdV and sine-Gordon equations}

For the modified KdV equation
\begin{equation}\label{mKdV}
 u_t=u_{xxx}+6u^2u_x,
\end{equation}
the recursion operator is $R= D^2+4u^2+4u_xD^{-1}u$; the negative symmetry is determined by equation $R(u_z)=\alpha u_z$, that is
\[
 u_{xxz}+4u^2u_z+4u_xD^{-1}(uu_z)=\alpha u_z.
\]
As in the previous examples, we consider the potential form of the equation
\begin{equation}\label{pmKdV}
 v_t=v_{xxx}+2v^3_x
\end{equation}
and also introduce the auxiliary variable $q$ according to the substitutions
\[
 u=v_x,\quad u_z=q_x/u=v_{xz}.
\]
This gives
\[
 v_{xxxz}+4uq_x+4u_xq=\alpha v_{xz} \quad\Rightarrow\quad v_{xxz}+4v_xq=\alpha v_z+\beta;
\]
then multiplying by $2v_{xz}$ and integrating once again brings to
\[
 2v_{xz}v_{xxz}+8qq_x=2(\alpha v_z+\beta)v_{xz} \quad\Rightarrow\quad 
 v^2_{xz}+4q^2=\alpha v^2_z+2\beta v_z+\gamma.
\]
Eliminating $q$ from these equations gives the negative symmetry for (\ref{pmKdV}):
\begin{equation}\label{pmKdV.vxxz}
 v_{xxz} = 2v_x\sqrt{\alpha v^2_z+2\beta v_z+\gamma-v^2_{xz}}+\alpha v_z+\beta.
\end{equation}
The following statement can be proved by direct calculations.

\begin{proposition}\label{pr:pmKdV}
Equations
\begin{equation}\label{pmKdV.vxxzi}
 v_{xxz_i} = 2v_x\sqrt{\alpha_iv^2_{z_i}+2\beta_iv_{z_i}+\gamma_i-v^2_{xz_i}}+\alpha_iv_{z_i}+\beta_i,
\end{equation}
are consistent with equation (\ref{pmKdV}) and 3D-consistent (for $\alpha_i\ne\alpha_j$) with equations
\begin{equation}\label{pmKdV.vzz}
 v_{z_iz_j}= \frac{2}{\alpha_j-\alpha_i}\left(
    v_{xz_j}\sqrt{\alpha_iv^2_{z_i}+2\beta_iv_{z_i}+\gamma_i-v^2_{xz_i}}
   -v_{xz_i}\sqrt{\alpha_jv^2_{z_j}+2\beta_jv_{z_j}+\gamma_j-v^2_{xz_j}}\right).
\end{equation}
\end{proposition}

Unlike previous examples, checking the identities $(v_{z_iz_j})_{z_k}=(v_{z_iz_k})_{z_j}$ requires taking into account the equations (\ref{pmKdV.vxxzi}).

If $\alpha=0$ then the negative symmetry is simplified: one can check that in this case the original variable $u=v_x$ satisfies the equation
\[
 uu_{xxz}-u_xu_{xz}+4u^3u_z+\beta u_x=0,
\]
which is consistent with (\ref{mKdV}). If also $\beta=0$, then the negative symmetry allows further degeneration to a hyperbolic equation: let us denote the corresponding independent variable by $y$ and set $\gamma=1$ (without loss of generality), then (\ref{pmKdV.vxxz}) takes the form
\[
 v_{xxy}= 2v_x\sqrt{1-v^2_{xy}}
\]
and it is easy to verify that the sine-Gordon equation
\begin{equation}\label{vxy}
 v_{xy} = \sin2v
\end{equation}
defines {\em special} solutions of this equation.  It can also be verified that consistency with other negative symmetries is maintained provided that for them $\alpha\ne0$ and $\beta=0$. As a result, the equation (\ref{vxy}) forms a compatible triple with the equations
\begin{gather*}
 v_{xxz}= 2v_x\sqrt{\alpha v^2_z+\gamma-v^2_{xz}}+\alpha v_z,\\
 v_{yz}= \frac{2}{\alpha}\left(\cos 2v\:v_{xz}-\sin2v\sqrt{\alpha v^2_z+\gamma-v^2_{xz}}\right).
\end{gather*}

\section{Negative symmetries from the lattice equations}\label{s:lattice}

\subsection{General scheme}\label{s:uxzzn}

The formula (\ref{neg}) with the recursion operator is not the only way to derive equations of the form (\ref{uxxz}). An alternative approach is associated with compatible pairs of differential-difference equations, of the dressing chain type
\begin{equation}\label{uxn}
 a(u_n,u_{n,x},u_{n+1},u_{n+1,x};\alpha)=0
\end{equation}
and of the Volterra lattice type
\begin{equation}\label{uzn}
 u_{n,z} = b(u_{n-1},u_n,u_{n+1}).
\end{equation}
The $\alpha$ parameter in all formulas in this section plays the same role as the parameter in (\ref{neg}), but does not necessarily coincide with it. Compatibility means that differentiating (\ref{uxn}) with respect to $z$ in virtue of (\ref{uzn}) gives an identity in virtue of the chain (\ref{uxn}) itself. Such pairs of chains have been studied in the literature for a long time, see e.g. \cite{Yamilov_1990, Garifullin_Habibullin_Yamilov_2015, Garifullin_Habibullin_2021} where examples and a number of classification results are given. It was also noted in these works that eliminating the variables $u_{n\pm1}$ leads to equations of the type (\ref{uxxz}). 

\begin{proposition}
Let equations (\ref{uxn}) and (\ref{uzn}) be compatible and satisfy the non-degeneracy conditions $\partial a/\partial u_{n,x}\ne0$, $\partial a/\partial u_{n+1,x}\ne0$, $\partial b/\partial u_{n\pm1}\ne0$. Then the variable $u=u_n$ satisfies, for any $n$, some  equation of type (\ref{uxxz}). 
\end{proposition}

\begin{proof}
Due to the non-degeneracy conditions, equation (\ref{uxn}) and its copy for $n=n-1$ can be solved with respect to $u_{n\pm1,x}$:
\[
 u_{n+1,x}=A^+(u_{n,x},u_n,u_{n+1}),\quad u_{n-1,x}=A^-(u_{n,x},u_{n-1},u_n).
\]
Then differentiation (\ref{uzn}) with respect to $x$ gives a relation of the form
\[
 u_{n,xz}=h(u_{n,x},u_{n-1},u_n,u_{n+1};\alpha).
\]
From here and from (\ref{uzn}), the variable $u_{n+1}$ is expressed as a function of $u_n,u_{n,x},u_{n,z},u_{n,xz}$, and we arrive at an equation of the form (\ref{uxxz}) by its substitution to (\ref{uxn}).
\end{proof}

A dressing chain (\ref{uxn}) is a hyperbolic equation with the continuous variable $x$ and the discrete variable $n$. The evolution symmetries of this equation fall into two subalgebras: one includes the chain (\ref{uzn}) and its higher symmetries, and the second contains a KdV type equation
\begin{equation}\label{utsep}
 u_t=u_{xxx}+f(u,u_x,u_{xx})
\end{equation}
and its higher symmetries. In other words, equations of type (\ref{uxn}) describe the $x$-part of the B\"acklund transformations for equations of type (\ref{utsep}). Note that B\"acklund transformations in the form (\ref{uxn}) exist only for such equations (\ref{ut}) in which the derivative $u_{xxx}$ appears linearly with a constant coefficient. For equations with more complex occurrences of $u_{xxx}$, additional transformations like $x\leftrightarrow u$ are required, which complicates the construction. In particular, for the Dym equation (\ref{Dym.ut}) this method is not directly applicable.

The commutativity property of B\"acklund transformations corresponding to different parameters $\alpha_i$ leads to 3D-consistent quad-equations
\begin{equation}\label{quad}
 F(u,T_i(u),T_j(u),T_iT_j(u);\alpha_i,\alpha_j)=0
\end{equation}
where $T_i:n_i\mapsto n_i+1$ denote shifts along discrete variables $n_i$ that form a multidimensional integer lattice. Each coordinate $n_i$ corresponds to a continuous variable $z_i$, a parameter $\alpha_i$ and a consistent pair of equations (\ref{uxn}) and (\ref{uzn}):
\begin{equation}\label{uTxz}
 a(u,u_x,T_i(u),T_i(u_x);\alpha_i)=0,\quad u_{z_i}= b(T^{-1}_i(u),u,T_i(u)).
\end{equation}
Then the elimination of the shifted variables, as described above, leads to equations (\ref{uxxzi}), and their consistency becomes a corollary of the consistency of equations on the lattice.

The consistency of equations of the Volterra lattice type with quad-equations (\ref{quad}) has also been studied in a lot of works. In \cite{Nijhoff_Papageorgiou_1991, Nijhoff_Walker_2001, Nijhoff_Ramani_Grammaticos_Ohta_2001} such equations were used for derivation of PDEs (with respect to variables $z_i$, $z_j$ in our notation) and their Painlev\'e type reductions. This topic was also developed in articles \cite{Tongas_Tsoubelis_Xenitidis_2001, Tsoubelis_Xenitidis_2009, Xenitidis_2011} where the symmetries of the Volterra lattice type (including non-autonomous ones) were systematically studied for the list of quad-equations from \cite{Adler_Bobenko_Suris_2003}. Generalizations to equations of higher order with respect to shifts are also known \cite{Adler_Postnikov_2014, Xenitidis_2018a, Xenitidis_2018b}. The papers \cite{Levi_Petrera_Scimiterna_Yamilov_2008, Levi_Yamilov_2009, Garifullin_Gudkova_Habibullin_2011, Garifullin_Yamilov_2012,  Garifullin_Mikhailov_Yamilov_2014} are devoted to the classification problem of quad-equations of general form, based on the existence of symmetries of the Volterra or Bogoyavlensky latttice type. Thus, compatible systems of equations (\ref{quad}) and (\ref{uTxz}) have been studied quite well. However, to the best of the author's knowledge, their connection with the  3D-consistency of negative flows (\ref{uxxzi}) have not been discussed previously.

\subsection{Simplest examples}

Let us illustrate the described scheme using the example of the pot-KdV equation. For it, the B\"acklund transformation is determined by the dressing chain
\begin{equation}\label{KdV.vxn}
 v_{n+1,x}+v_{n,x}= (v_{n+1}-v_n)^2+\alpha
\end{equation}
and one can check that this equation is consistent with 
\begin{equation}\label{KdV.vzn}
 v_{n,z}= \frac{1}{v_{n+1}-v_{n-1}}.
\end{equation}
In order to eliminate $v_{n\pm1}$, we use the relation
\[
 v_{n,xz}= -\frac{v_{n+1,x}-v_{n-1,x}}{(v_{n+1}-v_{n-1})^2}
  = -\frac{(v_{n+1}-v_n)^2-(v_n-v_{n-1})^2}{(v_{n+1}-v_{n-1})^2}
  = -\frac{v_{n+1}-2v_n+v_{n-1}}{v_{n+1}-v_{n-1}}.
\]
From here and from the previous equation we find
\[
 v_{n+1}=v_n-\frac{v_{n,xz}-1}{2v_{n,z}}
\]
and substitution into (\ref{KdV.vxn}) gives, after simple manipulations, the following equation for $v=v_n$:
\[
 2v_zv_{xxz}=v^2_{xz}+4(2v_x-\alpha)v^2_z-1,
\]
which coincides with (\ref{KdV.vz}) for $\beta=1$. It is easy to see that equation with an arbitrary $\beta\ne0$ is obtained by scaling of $z$ which amounts to multiplying the right-hand side of (\ref{KdV.vzn}) by $\sqrt{\beta}$ . 

As a result, it turns out that the negative symmetries of pot-KdV (\ref{KdV.zi}) and the additional equation (\ref{KdV.vzz}) are obtained by eliminating shifts from the chains
\[
  T_i(v_x)+v_x= (T_i(v)-v)^2+\alpha_i,\quad v_{z_i}= \frac{\sqrt{\beta_i}}{T_i(v)-T^{-1}_i(v)}
\]
and from the quad-equation $H_1$ (we use the notation for quad-equations from \cite{Adler_Bobenko_Suris_2003})
\[
 (v-T_iT_j(v))(T_i(v)-T_j(v))=\alpha_i-\alpha_j.
\]
The consistency of these equations on the lattice is easily verified (and this is a known result), from which the Proposition \ref{pr:KdV} about the 3D-consistency of equations (\ref{KdV.zi}) follows. However, it should be noted that in this scheme the value $\beta=0$ is omitted, which is not distinguished in the approach with the recursion operator where this parameter plays the role of an arbitrary integration constant.

For the Schwarzian-KdV equation, quite similarly, the Proposition \ref{pr:S-KdV} about the negative symmetry (\ref{S-KdV.zi}) and additional equations (\ref{S-KdV.uzz}) is derived by eliminating shifts from the chains
\[
 T_i(u_x)u_x= \alpha_i(T_i(u)-u)^2,\quad 
 u_{z_i}= \sqrt{\beta_i}\,\frac{(T_i(u)-u)(u-T^{-1}_i(u))}{T_i(u)-T^{-1}_i(u)}
\]
and from the quad-equation $Q_1(0)$
\[
 \alpha_i(u-T_j(u))(T_i(u)-T_iT_j(u))=\alpha_j(u-T_i(u))(T_j(u)-T_iT_j(u)).
\]
The pot-mKdV equation with the minus sign $v_t=v_{xxx}-2v^3_x$ corresponds to the chains
\[
 T_i(v_x)+v_x= a_i\cosh(T_i(v)-v),\quad v_{z_i}= c_i\frac{e^{T_i(v)}+e^{T^{-1}_i(v)}}{e^{T_i(v)}-e^{T^{-1}_i(v)}}-b_i
\]
for which the elimination of shifts leads to the negative symmetry
\[
 v_{xxz_i}= 2v_x\sqrt{v^2_{xz_i}+a^2_i(v_{z_i}+b_i)^2-a^2_ic^2_i}-a^2_i(v_{z_i}+b_i),
\]
which coincides with (\ref{pmKdV.vxxzi}) up to the change $v\to v\sqrt{-1}$ and denoting the parameters, and the corresponding quad-equation becomes $H_3(0)$ for the variables $q=e^v$:
\[
 a_i(qT_i(q)+T_j(q)T_iT_j(q))=a_j(qT_j(q)+T_i(q)T_iT_j(q)).
\]
For the pot-mKdV equation with the plus sign (\ref{pmKdV}), the real form of quad-equation is obtained by another change $q=\tan(v/2)$.

\subsection{Krichever--Novikov equation}\label{s:KN}

The most complicated example is related with the Krichever--Novikov equation \cite{Krichever_Novikov_1980}
\begin{equation}\label{KN}
 u_t=u_{xxx}-\frac{3(u^2_{xx}-r(u))}{2u_x},\quad r=c_4u^4+c_3u^3+c_2u^2+c_1u+c_0.
\end{equation}
Let us recall that linear fractional transformations preserve the form of this equation, but change the polynomial $r$, which makes possible to reduce it to one of the canonical forms, depending on the multiplicity of its zeroes. In the case of multiple zeroes, differential substitutions are known that connect (\ref{KN}) with the KdV equation, but if all the zeroes are simple, such a substitution does not exist \cite{Svinolupov_Sokolov_Yamilov_1983}. The recursion operator \cite{Sokolov_1984, Demskoi_Sokolov_2008} in the case of simple zeroes has the minimal order equal to 4 and it generates the higher symmetries starting from two seed symmetries: $u_{t_0}=u_x$ and equation (\ref{KN}) itself. From here it is clear that the formula (\ref{neg}) gives an equation of higher order with respect to the derivatives comparing to (\ref{uxxz}), but it turns out that an equation like (\ref{uxxz}) arises from the reduction found in \cite{Adler_2024b} using a method involving squared eigenfunctions for the Lax representation. Up to some changes\footnote{The correspondence with equations (58)--(60) from \cite{Adler_2024b} is given by formulas $-2u=v$, $q=2\beta g/v_x$, $-2\alpha=\beta$, $-2\beta=\gamma$ and $\partial_z=-\beta\partial_z$ with the present variables in the left hand side and old variables in the right hand side.}, this special negative symmetry is given in the following proposition, which is proven by direct calculation.

\begin{proposition}\label{pr:KNz}
Let $u$ satisfy equation (\ref{KN}) with $r=4u^3-g_2u-g_3$ and $(\alpha,\beta)$ is a parameter on the curve $\beta^2=r(\alpha)$. Then: 1) equations
\begin{gather}
\label{KN.uq}
 2qq_{xx}-q^2_x+2qq_x\left(\frac{u_{xx}}{u_x}-\frac{u_x}{u-\alpha}\right)
   -q^2\left(\frac{r(u)}{u^2_x}-\frac{2\beta}{u-\alpha}\right)
  +\frac{(u-\alpha)^2}{u^2_x}=0,\\
\label{KN.qt}
 q_t= -q_x\left(\frac{u_{xxx}}{u_x}-\frac{u^2_{xx}-r(u)}{2u^2_x}
  -\frac{2(u_{xx}-\beta)}{u-\alpha}\right)
  -\frac{q}{u_x}\left(\frac{2(\beta u_{xx}-r(u))}{u-\alpha}+r'(u)\right)
\end{gather}
are consistent, that is, this pair of equations defines a prolongation of equation (\ref{KN}) onto the variable $q$; 2) the flow
\begin{equation}\label{KN.uz}
 u_z= \gamma\left(\frac{u^2_xq^2_x-r(u)q^2}{(u-\alpha)q}-\frac{u-\alpha}{q}\right)
\end{equation}
is consistent with (\ref{KN}).
\end{proposition}

The variable $q$ can be eliminated from equations (\ref{KN.uq}) and (\ref{KN.uz}), which brings the negative symmetry to the form (\ref{uxxz}): 
\begin{equation}\label{KN.uxxz}
\begin{aligned}
 &P(u)(u_xu_{xxz}-u_{xx}u_{xz})^2\\ 
 &\qquad -u^2_x\bigl(P'(u)u_{xz}-(4u^2-8\alpha u-8\alpha^2+g_2)u_xu_z\bigr)(u_xu_{xxz}-u_{xx}u_{xz})\\
 &\qquad +(2\beta u^2_x-P(u))\bigl(r(u)u_{xz}^2-r'(u)u_xu_zu_{xz}+4(2u+\alpha)u^2_xu^2_z\bigr)\\
 &\qquad +4u^2_x((u-\alpha)u_{xz}-u_xu_z)^2-16\gamma^2u^2_x(\beta u^2_x-P(u))^2=0
\end{aligned}
\end{equation}
where $P(u)=u^4+\frac{1}{2}g_2u^2+2g_3u+\frac{1}{16}g^2_2-\alpha r(u)$.

The cumbersome form of this equation makes the direct verification of its 3D-consistency very difficult. However, this can be done in other variables by using the representation of the negative flow by the pair of chains
\begin{gather}
\label{KN.nx}
 u_{n,x}u_{n+1,x}=h(u_n,u_{n+1}),\\
\label{KN.nz}
 u_{n,z}=f(u_{n-1},u_n,u_{n+1})=\frac{2h(u_n,u_{n+1})}{u_{n+1}-u_{n-1}}-h^{(0,1)}(u_n,u_{n+1})
\end{gather}
where $h(u,v)$ is a symmetric biquadratic polynomial, that is
\begin{equation}\label{KN.h}
 h^{(3,0)}(u,v)=h^{(0,3)}(u,v)=0,\quad h(u,v)=h(v,u).
\end{equation}
The chain (\ref{KN.nx}) defines the $x$-part of the B\"acklund transformation for equation (\ref{KN}) with the polynomial
\begin{equation}\label{rh}
 r(u)=h^{(0,1)}(u,v)^2-2h(u,v)h^{(0,1)}(u,v),
\end{equation}
and the chain (\ref{KN.nz}) is the equation $V_4(0)$ in the Yamilov's classification \cite{Yamilov_1983, Yamilov_2006}, which defines a discretization of the Krichever--Novikov equation, and is also related to the B\"acklund transformation for the Landau--Lifshitz equation \cite{Shabat_Yamilov_1991, Levi_Petrera_Scimiterna_Yamilov_2008}. It is known \cite{Adler_1998, Adler_Suris_2004, Levi_Petrera_Scimiterna_Yamilov_2008} that each of the chains (\ref{KN.nx}) and (\ref{KN.nz}) defines a continuous symmetry for the multi-dimensional lattice governed by $Q_4$ quad-equation. Therefore, all that needs to be done is to check the consistency of these two chains with each other and make sure that eliminating $u_{n\pm1}$ leads to the equation (\ref{KN.uxxz}), for a suitable choice of polynomial $h$. The following statement is easy to prove in general form, using only the properties (\ref{KN.h}).

\begin{proposition}
Equations (\ref{KN.nx}) and (\ref{KN.nz}) are consistent.
\end{proposition}

\begin{proof}
We denote $h_n=h(u_n,u_{n+1})$ and $f_n=f(u_{n-1},u_n,u_{n+1})$. From (\ref{KN.h}) it follows that $f_n$ can be expressed in terms of $h_{n-1}$ as well:
\begin{equation}\label{fhh}
 f_n=\frac{2h_n}{u_{n+1}-u_{n-1}}-h^{(0,1)}_n =\frac{2h_{n-1}}{u_{n+1}-u_{n-1}}+h^{(1,0)}_{n-1}. 
\end{equation}
The coincidence of two expressions follows from the Taylor expansion
\[
 h(u_n,u_{n+1})=h(u_n,u_{n-1})+h^{(0,1)}(u_n,u_{n-1})(u_{n+1}-u_{n-1})
  +\frac{1}{2}h^{(0,2)}(u_n,u_{n-1})(u_{n+1}-u_{n-1})^2
\]
after alternating $u_{n-1}$ and $u_{n+1}$. The consistency of (\ref{KN.nx}) and (\ref{KN.nz}) means that the following equality is fulfilled identically in virtue of (\ref{KN.nx}):
\begin{equation}\label{cons}
 D(f_n)u_{n+1,x}+u_{n,x}D(f_{n+1}) = h^{(1,0)}_nf_n+h^{(0,1)}_nf_{n+1}.
\end{equation}
We have
\begin{equation}\label{KN.nxz}
u_{n,xz}=D(f_n)= f^{(1,0,0)}_nu_{n-1,x}+f^{(0,1,0)}_nu_{n,x}+f^{(0,0,1)}_nu_{n+1,x}=f^{(0,1,0)}_nu_{n,x}, 
\end{equation}
since the first and third terms are canceled due to the consequences of (\ref{fhh}):
\[
 f^{(1,0,0)}_n=\frac{2h_n}{(u_{n+1}-u_{n-1})^2},~~
 f^{(0,0,1)}_n=-\frac{2h_{n-1}}{(u_{n+1}-u_{n-1})^2},~~
 u_{n-1,x}=\frac{h_{n-1}}{u_{n,x}},~~ u_{n+1,x}=\frac{h_n}{u_{n,x}}.
\]
Therefore (\ref{cons}) amounts to
\[
 \bigl(f^{(0,1,0)}_n+f^{(0,1,0)}_{n+1}\bigr)h_n = h^{(1,0)}_nf_n+h^{(0,1)}_nf_{n+1},
\]
and this is equivalent to the identity:
\begin{gather*}
 \Bigl(\frac{2h^{(1,0)}_n}{u_{n+1}-u_{n-1}}-h^{(1,1)}_n 
      +\frac{2h^{(0,1)}_n}{u_{n+2}-u_n}+h^{(1,1)}_n\Bigr)h_n\\
 \qquad= h^{(1,0)}_n\Bigl(\frac{2h_n}{u_{n+1}-u_{n-1}}-h^{(0,1)}_n\Bigr)
  +h^{(0,1)}_n\Bigl(\frac{2h_n}{u_{n+2}-u_n}+h^{(1,0)}_n\Bigr).
\qedhere
\end{gather*}
\end{proof}

Next, let us denote $u=u_n$ and $v=u_{n+1}$, and eliminate $u_{n-1}$ from equations (\ref{KN.nz}) and (\ref{KN.nxz}). It is easy to check that the resulting system is of the form (cf. with \cite{Adler_Shabat_2012})
\begin{equation}\label{uvxz}
\begin{aligned}
 & u_xv_x=h,\quad h=h(u,v),\\  
 & hu_{xz}-h^{(1,0)}u_xu_z+u_x(hh^{(1,1)}-h^{(0,1)}h^{(1,0)})=0. 
\end{aligned}
\end{equation}
The variable $v$ is algebraically expressed from the second equation, then the substitution into the first one gives some equation of the form (\ref{uxxz}). This step must be done by passing to a specific form of the polynomial $h$ related to $r=4u^3-g_2u-g_3$ by formula (\ref{rh}):
\begin{equation}\label{h-Weierstrass}
 h(u,v)=\frac{1}{\nu}((uv+\mu u+\mu v+g_2/4)^2-(u+v+\mu)(4\mu uv-g_3)),\quad \nu^2=r(\mu).
\end{equation}
The points $(\mu,\nu)$ and $(\alpha,\beta)$ are different, although they lie on the same algebraic curve. Straightforward, rather tedious calculations prove that the system (\ref{uvxz}), (\ref{h-Weierstrass}) is equivalent to (\ref{KN.uxxz}) under the choice
\begin{gather*}
 \alpha=\frac{16\mu^4+8g_2\mu^2+32g_3\mu+g^2_2}{16r(\mu)},\\
 \beta=\frac{\nu(64\mu^6-80g_2\mu^4-320g_3\mu^3-20g^2_2\mu^2-16g_2g_3\mu-32g^2_3)}{32r(\mu)^2}
\end{gather*}
and $\gamma=\pm\frac{1}{2}$ (this parameter can be changed by scaling of $z$).

\section{Conclusion}

In our study, equations of the Volterra lattice type (\ref{uzn}) played an auxiliary role as a tool for constructing negative flows for equations of the KdV type (\ref{ut}). However, these equations can also be considered as the main ones. For example, the negative symmetry for the Volterra lattice
\[
 u_{n,t}=u_n(u_{n+1}-u_{n-1})
\]
was obtained by the recursion operator method in \cite{Adler_2024a}. In the potential variables given by $u_n=v_{n,t}=e^{v_{n+1}-v_{n-1}}$, it takes the form
\[
 e^{v_{n+1}-v_{n-1}}(v_{n+1,z}+v_{n,z})(v_{n,z}+v_{n-1,z})=\alpha v^2_{n,z}+(-1)^n\beta v_{n,z}+\gamma,
\]
which can be viewed as a discrete analog of type (\ref{uxxz}) equation. This equation satisfies the 3D-consistency property in the sense analogous to the definition from section \ref{s:def}, with replacement of derivatives with respect to $x$ by shifts with respect to $n$. The results from section \ref{s:lattice} also have their analogues for other Volterra type lattices, which are planned to be presented in future publication.

Similar results on the 3D-consistency of negative symmetries can be obtained for some other classes of equations. The method of constructing negative symmetries by use of auxiliary differential-difference equations also allows generalizations (for example, for systems like the nonlinear Schr\"odinger equation, pairs of Toda and relativistic Toda type lattices can be used), but its equivalence to the direct method based on the recursion operator remains an open question. It should be noted that for equations associated with spectral problems of order higher than 2, the form of recursion operators, and hence negative symmetries, becomes noticeably more complicated (examples associated with the Boussinesq equation and the Drinfeld--Sokolov system were studied in \cite{Adler_2024b}).

\subsection*{Conflict of interest} 

The author has no conflicts to disclose.



\begin{thebibliography}{99}

\bibitem{Kamchatnov_Pavlov_2002} A.M. Kamchatnov, M.V. Pavlov. 
 On generating functions in the AKNS hierarchy. {\em Phys. Lett. A \bf 301:3--4} (2002)
 \href{https://doi.org/10.1016/S0375-9601(02)00935-0}{269--274}.

\bibitem{Aratyn_Gomes_Zimerman_2006} H. Aratyn, J.F. Gomes, A.H. Zimerman.
 On negative flows of the AKNS hierarchy and a class of deformations of a bihamiltonian structure of hydrodynamic type. 
 {\em J. Phys. A: Math. Gen. \bf 39:5} (2006) \href{https://doi.org/10.1088/0305-4470/39/5/006}{1099--1114}.

\bibitem{Adans_Franca_Gomes_Loboa_Zimerman_2023} Y.F. Adans, G. Fran\c{c}a, J.F. Gomes, G.V. Loboa, A.H. Zimerman.
Negative flows of generalized KdV and mKdV hierarchies and their gauge-Miura transformations.
{\em JHEP \bf 08} (2023) \href{https://doi.org/10.1007/JHEP08(2023)160}{160}.

\bibitem{Lou_Jia_2024} S.Y. Lou, M. Jia. From one to infinity: symmetries of integrable systems.
{\em JHEP \bf 02} (2024) \href{https://doi.org/10.1007/JHEP02(2024)172}{172}.

\bibitem{Adler_2024a} V.E. Adler. Negative flows and non-autonomous reductions of the Volterra lattice. 
{\em Open Communications in Nonl. Math. Phys., Special Issue in Memory of Decio Levi} (2024)
 \href{https://doi.org/10.46298/ocnmp.11597}{11597}.

\bibitem{Adler_2024b} V.E. Adler. Negative flows for several integrable models.  
{\em J. Math. Phys. \bf 65} (2024) \href{https://doi.org/10.1063/5.0181692}{023502}.

\bibitem{Schiff_1998} J. Schiff. The Camassa--Holm equation: a loop group approach.
 {\em Physica D \bf 121:1--2} (1998) \href{https://doi.org/10.1016/S0167-2789(98)00099-2}{24--43}.

\bibitem{Hone_1999} A.N.W. Hone. The associated Camassa--Holm equation and the KdV equation.
 {\em J. Phys. A: Math. Gen. \bf 32:27} (1999) \href{https://doi.org/10.1088/0305-4470/32/27/103}{L307--314}.

\bibitem{Rogers_Schief_2002} C. Rogers, W.K. Schief. B\"acklund and Darboux transformations. 
 Geometry and modern applications in soliton theory. Cambridge University Press, Cambridge, 2002.

\bibitem{Aratyn_Gomes_Zimerman_2006b} H. Aratyn, J.F. Gomes, A.H. Zimerman.
 On a negative flow of the AKNS hierarchy and its relation to a two-component Camassa--Holm equation.
 {\em SIGMA \bf 2} (2006) \href{https://doi.org/10.3842/SIGMA.2006.070}{070}.

\bibitem{Meshkov_Sokolov_2011} A.G. Meshkov, V.V. Sokolov.
 Hyperbolic equations with third-order symmetries. {\em Theor. Math. Phys. \bf 166:1} (2011)
 \href{https://doi.org/10.1007/s11232-011-0004-3}{43--75}.

\bibitem{Orlov_Rauch-Wojciechowski_1993} A.Yu. Orlov, S. Rauch-Wojciechowski. 
 Dressing method, Darboux transformation and generalized restricted flows for the KdV hierarchy. 
 {\em Physica D \bf 69:1--2} (1993) \href{https://doi.org/10.1016/0167-2789(93)90181-Y}{77--84}.

\bibitem{Adler_Kolesnikov_2023} V.E. Adler, M.P. Kolesnikov. 
 Non-autonomous reductions of the KdV equation and multi-component analogs of the Painlev\'e equations P$_{34}$ and P$_3$. 
 {\em J. Math. Phys. \bf 64} (2023) \href{https://doi.org/10.1063/5.0156409}{101505}.

\bibitem{Ferapontov_1997} E.V. Ferapontov. 
 Laplace transformations of hydrodynamic-type systems in Riemann invariants: periodic sequences. 
 {\em J. Phys. A: Math. Gen. \bf 30:19} (1997) \href{https://doi.org/10.1088/0305-4470/30/19/023}{6861--6878}.

\bibitem{Nijhoff_Walker_2001} F.W. Nijhoff, A.J. Walker. The discrete and continuous Painlev\'e VI hierarchy and the Garnier system. {\em Glasgow Math. J. \bf 43A} (2001) \href{https://doi.org/10.1017/S0017089501000106}{109--123}.

\bibitem{Adler_Bobenko_Suris_2003} V.E. Adler, A.I. Bobenko, Yu.B. Suris.
 Classification of integrable equations on quad-graphs. The consistency approach.
 {\em Comm. Math. Phys. \bfseries 233:3} (2003) \href{https://doi.org/10.1007/s00220-002-0762-8}{513--543}. 
 
\bibitem{Adler_Bobenko_Suris_2009} V.E. Adler, A.I. Bobenko, Yu.B. Suris.
 Discrete nonlinear hyperbolic equations. Classification of integrable cases.
 {\em Funct. Anal. and Appl. \bfseries 43:1} (2009) \href{https://doi.org/10.1007/s10688-009-0002-5}{3--17}. 

\bibitem{Adler_Postnikov_2014} V.E. Adler, V.V. Postnikov. On discrete 2D integrable equations of higher order
 {\em J. Phys. A: Math. Theor. \bf 47:4} (2014) \href{https://doi.org/10.1088/1751-8113/47/4/045206}{045206}. 
 
\bibitem{Xenitidis_2019} P.D. Xenitidis. On consistent systems of difference equations.
 {\em J. Phys. A: Math. Theor. \bf 52:45} (2019) \href{https://doi.org/10.1088/1751-8121/ab48b0}{455201}.

\bibitem{Yamilov_1990} R.I. Yamilov. Invertible changes of variables generated by B\"acklund transformations. {\em Theor. Math. Phys. \bf 85:3} (1990) \href{https://doi.org/10.1007/BF01018403}{1269--1275}. 

\bibitem{Gelfand_Dikii_1975} I.M. Gel'fand, L.A. Dikii. 
 Asymptotic properties of the resolvent of Sturm--Liouville equations, and the algebra of Korteweg--de Vries equations. 
 {\em Russian Math. Surveys \bf 30:5} (1975) \href{https://doi.org/10.1070/rm1975v030n05abeh001522}{77--113}.

\bibitem{Alonso_Shabat_2004} L. Mart\'\i nez Alonso, A.B. Shabat. Hydrodynamic reductions and solutions of the universal hierarchy.
 {\em Theor. Math. Phys. \bf 140:2} (2004) \href{https://doi.org/10.1023/B:TAMP.0000036538.41884.57}{1073--1085}.
 
\bibitem{Adler_Shabat_2007} V.E. Adler, A.B. Shabat. Model equation of the theory of solitons.
 {\em Theor. Math. Phys. \bfseries 153:1} (2007) \href{https://doi.org/10.1007/s11232-007-0121-1}{1373--1387}.

\bibitem{Dorfman_1987} I.Ya. Dorfman. Dirac structures of integrable evolution equations.
 {\em Phys. Lett. A \bf 125:5} (1987) \href{https://doi.org/10.1016/0375-9601(87)90201-5}{240--246}.

\bibitem{Wang_2002} J.P. Wang. A list of $1+1$ dimensional integrable equations and their properties.
 {\em J. Nonl. Math. Phys. \bf 9:1} (2002) \href{https://doi.org/10.2991/jnmp.2002.9.s1.18}{213--233}.

\bibitem{Garifullin_Habibullin_Yamilov_2015} R.N. Garifullin, I.T. Habibullin, R.I. Yamilov.
 Peculiar symmetry structure of some known discrete nonautonomous equations.
 {\em J. Phys. A: Math. Theor. \bf 48:23} (2015) \href{https://doi.org/10.1088/1751-8113/48/23/235201}{235201}.

\bibitem{Garifullin_Habibullin_2021} R.N. Garifullin, I.T. Habibullin.
 Generalized symmetries and integrability conditions for hyperbolic type semi-discrete equations.
 {\em J. Phys. A: Math. Theor. \bf 54:20} (2021) \href{https://doi.org/10.1088/1751-8121/abf3ea}{205201}.

\bibitem{Nijhoff_Papageorgiou_1991} F.W. Nijhoff, V.G. Papageorgiou. 
 Similarity reductions of integrable lattices and discrete analogues of Painlev\'e PII equation. 
 {\em Phys. Lett. A \bf 153:6--7} (1991) \href{https://doi.org/10.1016/0375-9601(91)90955-8}{337--344}.

\bibitem{Nijhoff_Ramani_Grammaticos_Ohta_2001} F.W. Nijhoff, A. Ramani, B. Grammaticos, Y. Ohta. 
 On discrete Painlev\'e equations associated with the lattice KdV systems and the Painlev\'e VI equation. 
 {\em Studies in Appl. Math. \bf 106:3} (2001) \href{https://doi.org/10.1111/1467-9590.00167}{261--314}.

\bibitem{Tongas_Tsoubelis_Xenitidis_2001} A. Tongas, D. Tsoubelis, P. Xenitidis.  
 A family of integrable nonlinear equations of hyperbolic type. 
 {\em J. Math. Phys. \bf 42:12} (2001) 
 \href{https://doi.org/10.1063/1.1416488}{5762--5784}.

\bibitem{Tsoubelis_Xenitidis_2009} D. Tsoubelis, P. Xenitidis. 
 Continuous symmetric reductions of the Adler--Bobenko--Suris equations. 
 {\em J. Phys. A: Math. Theor. \bf 42:16} (2009) 
 \href{https://doi.org/10.1088/1751-8113/42/16/165203}{165203}. 

\bibitem{Xenitidis_2011} P.D. Xenitidis. 
 Symmetries and conservation laws of the ABS equations and corresponding differential-difference equations of Volterra type.
 {\em J. Phys. A: Math. Theor. \bf 44:43} (2011)
 \href{https://doi.org/10.1088/1751-8113/44/43/435201}{435201}. 

\bibitem{Xenitidis_2018a} P.D. Xenitidis. Determining the symmetries of difference equations.
 {\em Proc. R. Soc. A \bf 474} (2018) 
 \href{https://doi.org/10.1098/rspa.2018.0340}{20180340}.

\bibitem{Xenitidis_2018b} P.D. Xenitidis. Deautonomizations of integrable equations and their reductions.
 {\em J. of Integrable Systems \bf 3:1 } (2018) 
 \href{https://doi.org/10.1093/integr/xyy009}{xyy009}.

\bibitem{Levi_Petrera_Scimiterna_Yamilov_2008} D. Levi, M. Petrera, C. Scimiterna, R. Yamilov.
 On Miura transformations and Volterra-type equations associated with the Adler--Bobenko--Suris equations. 
 {\em SIGMA \bf 4} (2008)
 \href{http://www.emis.de/journals/SIGMA/2008/077/}{077}.

\bibitem{Levi_Yamilov_2009} D. Levi, R.I. Yamilov. 
 The generalized symmetry method for discrete equations.
 {\em J. Phys. A: Math. Theor. \bf 42} (2009) 
 \href{https://doi.org/10.1088/1751-8113/42/45/454012}{454012}.

\bibitem{Garifullin_Gudkova_Habibullin_2011} R.N. Garifullin, E.V. Gudkova, I.T. Habibullin.
 Method for searching higher symmetries for quad-graph equations.
 {\em J. Phys. A: Math. Theor. \bf 44} (2011) 
 \href{https://doi.org/10.1088/1751-8113/44/32/325202}{325202}.

\bibitem{Garifullin_Yamilov_2012} R.N. Garifullin, R.I. Yamilov.
 Generalized symmetry classification of discrete equations of a class depending on twelve parameters.
 {\em J. Phys. A: Math. Theor. \bf 45} (2012) 
 \href{https://doi.org/10.1088/1751-8113/45/34/345205}{345205}.
 
\bibitem{Garifullin_Mikhailov_Yamilov_2014} R.N. Garifullin, A.V. Mikhailov, R.I. Yamilov. 
 Discrete equation on a square lattice with a nonstandard structure of generalized symmetries.
 {\em Theor. Math. Phys. \bf 180:1} (2014) 
 \href{https://doi.org/10.1007/s11232-014-0178-6}{765--780}.

\bibitem{Krichever_Novikov_1980} I.M. Krichever, S.P. Novikov. 
 Holomorphic bundles over algebraic curves and nonlinear equations. 
 {\em Russian Math. Surv. \bf 35:6} (1980) 
 \href{http://doi.org/10.1070/RM1980v035n06ABEH001974}{53--79}.
 
\bibitem{Svinolupov_Sokolov_Yamilov_1983} S.I. Svinolupov, V.V. Sokolov, R.I. Yamilov. 
 On B\"acklund transformations for integrable evolution equations. 
 {\em Sov. Math. Dokl. \bfseries 28} (1983) 
 \href{http://mi.mathnet.ru/eng/dan46205}{165--168}.
 
\bibitem{Sokolov_1984} V.V. Sokolov. 
 Hamiltonian property of the Krichever--Novikov equation.
 {\em Doklady Akad. Nauk SSSR \bf 277:1} (1984) 
 \href{http://mi.mathnet.ru/eng/dan9572}{48--50}.

\bibitem{Demskoi_Sokolov_2008} D.K. Demskoi, V.V. Sokolov. 
 On recursion operators for elliptic models. 
 {\em Nonlinearity \bf 21} (2008) 
 \href{https://doi.org/10.1088/0951-7715/21/6/006}{1253--1264}.

\bibitem{Yamilov_1983} R.I. Yamilov. On classification of discrete evolution equations. 
 {\em Uspekhi Math. Nauk \bf 38:6} (1983) \href{http://mi.mathnet.ru/umn3034}{155--156} (in Russian).

\bibitem{Yamilov_2006} R.I. Yamilov. 
 Symmetries as integrability criteria for differential difference equations.
 {\em J. Phys. A \bf 39:45} (2006) 
 \href{https://doi.org/10.1088/0305-4470/39/45/R01}{R541--623}.

\bibitem{Shabat_Yamilov_1991} A.B. Shabat, R.I. Yamilov. Symmetries of nonlinear chains.
 {\em Leningrad Math. J. \bf 2:2} (1991) 377--399.

\bibitem{Adler_1998} V.E. Adler. B\"acklund transformation for the Krichever--Novikov equation.
 {\em Intl. Math. Res. Notices \bf 1998:1} (1998) 
 \href{https://doi.org/10.1155/S1073792898000014}{1--4}.

\bibitem{Adler_Suris_2004} V.E. Adler, Yu.B. Suris. 
 Q4: Integrable master equation related to an elliptic curve. 
 {\em Intl. Math. Res. Notices \bf} (2004) 
 \href{https://doi.org/10.1155/S107379280413273X}{2523--2553}.

\bibitem{Adler_Shabat_2012} V.E. Adler, A.B. Shabat.
 Toward a theory of integrable hyperbolic equations of third order.
 {\em J. Phys. A: Math. Theor. \bf 45} (2012) 
 \href{https://doi.org/10.1088/1751-8113/45/39/395207}{395207}.
\end{thebibliography}
\end{document}